\documentclass[letterpaper, 10 pt, conference]{ieeeconf}
\IEEEoverridecommandlockouts
\usepackage{bbold}
\usepackage{enumerate}
\usepackage{graphicx}
\usepackage{amsmath,amssymb}
\usepackage{color}
\usepackage{lipsum}
\usepackage{abraces}
\usepackage{tcolorbox}
\usepackage{hyperref}
\usepackage{xcolor}
\usepackage[noadjust]{cite}

\usepackage{amsmath,bm}
\usepackage{color}
\usepackage{multirow}
\usepackage{mathrsfs}
\usepackage{lineno,hyperref}
\usepackage{dsfont}
\usepackage{mathtools}
\usepackage{amssymb}
\usepackage{cite}
\usepackage{float}
\usepackage{stfloats}
\usepackage{ntheorem}
\theorembodyfont{\rmfamily}
\usepackage{amsmath,bm,bbm}
\usepackage{color}
\usepackage{subcaption}  
\usepackage{tikz,pgfplots}
\usepackage{mathtools}
\usepackage{algorithm}
\usepackage{algpseudocode}
\algrenewcommand\algorithmicrequire{\textbf{Input:}}
\algrenewcommand\algorithmicensure{\textbf{Output:}}

\newtheorem{remark}{{Remark}}
\newtheorem{lemma}{Lemma}
\newtheorem{definition}{Definition}
\newtheorem{theorem}{Theorem}
\newtheorem{assumption}{Assumption}
\newtheorem{corollary}{Corollary}

\newtheorem{example}{{Example}}
\newtheorem{problem}{Problem}
\newtheorem{condition}{Condition}

\def\begmat#1{\begin{bmatrix}#1\end{bmatrix}}

\def\cale{{\cal E}}

\def\cale{{\cal E}}

\def\cala{{\cal A}}

\def\dimx{n}
\def\dimu{m}
\def\dimy{p}
\def\colM{\mu}
\def\eps{\beta}
\def\vec{\mbox{vec}}

\def\L2e{{\cal L}_{2e}}

\def\rea{\mathbb{R}}

\def\diag{\mbox{diag}}

\def\col{\mbox{col}}

\def\diag{\mbox{diag}}
\def\rank{\mbox{rank}\;}
\def\im{\mbox{Im}\;}

\def\ker{\mbox{Ker}\;}

\def\span{{\mbox{span}}}

\usepackage[prependcaption,colorinlistoftodos]{todonotes}

\def\cala{{\cal A}}

\def\cale{{\cal E}}

\def\col{\mbox{col}}

\def\L2{{\cal L}_2}
\def\L2e{{\cal L}_{2e}}

\def\rea{\mathbb{R}}

\def\diag{\mbox{diag}}

\def\begequarr{\begin{eqnarray}}
\def\endequarr{\end{eqnarray}}
\def\begequarrs{\begin{eqnarray*}}
\def\endequarrs{\end{eqnarray*}}
\def\begarr{\begin{array}}
\def\endarr{\end{array}}
\def\begequ{\begin{equation}}
\def\endequ{\end{equation}}

\def\begdes{\begin{description}}
\def\enddes{\end{description}}
\def\begenu{\begin{enumerate}}
\def\begite{\begin{itemize}}
\def\endite{\end{itemize}}
\def\endenu{\end{enumerate}}

\def\lef[{\left[\begin{array}}
\def\rig]{\end{array}\right]}
\def\begcen{\begin{center}}
\def\endcen{\end{center}}

\def\QED{\hfill $\blacksquare$}

\def\begmat#1{\begin{bmatrix}#1\end{bmatrix}}

\usepackage{color}

\definecolor{myblue}{RGB}{0,114,178}
\definecolor{myorange}{RGB}{230,159,0}
\definecolor{mygreen}{RGB}{0,158,115}

\usepackage{xcolor}
\definecolor{codegreen}{RGB}{0,100,0}
\algrenewcommand\algorithmiccomment[1]{\hfill\textcolor{codegreen}{\ttfamily\% #1}}

\title{\LARGE \bf
Input–Output Data-Driven Stabilization of Continuous-Time Linear MIMO Systems
}

\author{Haihui Gao, Alessandro Bosso, Lei Wang, David Saussi\'e, and Bowen Yi
\thanks{This work was supported in part by the Natural Sciences and Engineering Research Council of Canada (NSERC) under Grants RGPIN-2020-06608 and RGPIN-2024-0478, the Programme PIED, the European Union's Horizon Europe research and innovation program under the Marie Sk{\l}odowska-Curie Grant Agreement No. 101104404 - \mbox{IMPACT4Mech}, and the National Natural Science Foundation of China under Grant 62573383. 
}
\thanks{H. Gao, D. Saussi\'e, and B. Yi are with the Department of Electrical Engineering, Polytechnique Montr\'eal and GERAD, Montr\'eal, Canada.
{\tt \{haihui.gao,d.saussie,bowen.yi\}@polymtl.ca}}
\thanks{A. Bosso is with Dept. of Electrical, Electronic, and Information Engineering, University of Bologna, Italy. {\tt alessandro.bosso@unibo.it}}
\thanks{L. Wang is with the College of Control Science and Engineering, Zhejiang University, Hangzhou, China. {\tt lei.wangzju@zju.edu.cn}}
}

\begin{document}

\maketitle
\thispagestyle{empty}
\pagestyle{empty}

\begin{abstract}

In this paper, we address the problem of data-driven stabilization of continuous-time multi-input multi-output (MIMO) linear time-invariant systems using the input–output data collected from an experiment.
Building on recent results for data-driven output-feedback control based on non-minimal realizations, we propose an approach that can be applied to a broad class of continuous-time MIMO systems without requiring a uniform observability index.
The key idea is to show that Kreisselmeier’s adaptive filter can be interpreted as an observer of a stabilizable non-minimal realization of the plant. 
Then, by postprocessing the input–output data with such a filter, we derive a linear matrix inequality that yields the feedback gain of a dynamic output-feedback stabilizer.
\end{abstract}

\section{Introduction}

Obtaining controllers of dynamical systems directly from trajectory data has a long history, ranging from adaptive control \cite{SASBOD} to modern data-driven approaches \cite{DataBasedCtrlBook}.
In particular, a major recent trend in data-driven control has been to design controllers directly from experimental datasets via linear matrix inequalities (LMIs).
This end-to-end paradigm avoids any intermediate modeling step and may be effective even when the data are not informative enough to perform system identification.

In the \emph{discrete-time} context, pioneering works on data-driven LMIs are \cite{DEPTES} and \cite{VANetal}, which mainly focus on the state-feedback stabilization of linear time-invariant (LTI) systems.
While the former derives LMIs by developing state-space results inspired by Willems et al.'s fundamental lemma \cite{WILetal}, the latter introduces the framework of data informativity to provide necessary and sufficient conditions for data-driven analysis and control.
Later, alternative LMIs particularly suited for noisy datasets have been derived from quadratic matrix inequalities \cite{van2023quadratic}.
The data-driven LMI paradigm has also been applied to time-varying systems \cite{10124991}, nonlinear systems \cite{DEPTES, van2023quadratic}, and output-feedback control, using either non-minimal realizations involving shifted input--output data \cite{DEPTES, ALSetal, LI2026112545} or the behavioral framework \cite{10122597}.

In contrast, direct data-driven control for \emph{continuous-time} systems remains less developed.
Recent progress includes continuous-time versions of the fundamental lemma \cite{LOPetal, SCHetal} 
and the study of the effect of sampling on informativity \cite{eising2024sampling}.
Unlike the discrete-time case, continuous-time methods must handle derivative signals, which are usually unavailable in practice.
Most state-feedback results derive LMIs assuming that state derivatives are measurable \cite{DEPTES, eising2024sampling, hu2025enforcing}.
To remove this need, \cite{rapisarda2023orthogonal} proposes orthogonal polynomial bases, while \cite{ohta2024sampling} explores more general sampling functionals.
On the other hand, data-driven control methods that rely purely on input--output data remain rare.
Notable exceptions include \cite{BOSetal24, BOSetal25}, which design output-feedback controllers for LTI systems using non-minimal realizations that can be reconstructed by linear filters of the input--output signals.
While these results address single-input single-output (SISO) and multi-input multi-output (MIMO) systems with a uniform observability index, respectively, finding an approach for general continuous-time MIMO LTI systems remains an open problem.
This paper aims to close this gap with the following contributions:

\begin{itemize}
\item[\bf C1] We propose a method to derive an output-feedback stabilizer directly from an input–output trajectory of MIMO LTI systems without restrictions on the observability indices.
We obtain this result by showing that Kreisselmeier's adaptive filter \cite{KRE79} acts as an observer of a canonical non-minimal realization of the plant \cite{BOSetal25}.
The resulting output-feedback controller combines the filter with a stabilizing feedback law depending on the filter states.

\item[\bf C2] To compute the gains of the feedback law, we postprocess the input--output data with the filter and construct an LMI with samples of the input, output, and filter trajectories.
Since the necessary rank conditions ensuring feasibility of the LMIs in \cite{BOSetal24, BOSetal25} do not hold in general due to the filter structure, we formulate the LMI applying a suitable decomposition of the data that exploits the stabilizability of the non-minimal realization, along with sufficient excitation of the controllable states.
\end{itemize}

\emph{Notation.}
We use $\rea$ and $\rea_-$ to denote the sets of real numbers and negative numbers, respectively.
$0_n$ and $0_{n\times m}$ denote the zero column vector of dimension $n$ and the zero matrix of dimension $n\times m$, respectively.
We use $|\cdot|$ to represent the Euclidean norm of a vector and $\otimes$ for the Kronecker product.
For a matrix $A\in \rea^{n\times m}$, the operator $\vec(A)$ denotes the 
$mn$-dimensional vector obtained by stacking the columns of $A$ one below another.

\section{Problem Statement}

Consider a continuous-time LTI MIMO system of the form
\begin{equation}
\label{eq:dot_x}
\begin{aligned}
    \dot{x} =  A x + Bu, \quad 
    y = C x
\end{aligned}
\end{equation}
where $x\in \rea^\dimx$ is the state, $u\in \rea^\dimu$ is the control input, and $y\in \rea^\dimy$ is the measured output, while $A$, $B$, and $C$ are constant unknown matrices of appropriate dimensions that satisfy the following assumption.

\begin{assumption}
\label{hyp:minimal}\it
    The pair $(A, B)$ is controllable and the pair $(C, A)$ is observable.
\end{assumption}

We aim at solving the following input--output data-driven stabilization problem.

\begin{problem}
\label{problem:1}\it
    Suppose that the input--output dataset
    \begin{equation}\label{eq:dataset}
        \cale_{\tt D}:=\{u(t),y(t)\}_{[0,t_{\tt D}]},
    \end{equation}
    has been acquired from the plant \eqref{eq:dot_x} during an experiment of duration $t_{\tt D} > 0$, under some excitation condition of the input $u$ and from an unknown initial condition $x(0) \in \rea^\dimx$.
    Moreover, we suppose that $u(t)$ and $x(t)$ are piecewise continuous and absolutely continuous over $[0, t_{\tt D}]$, respectively.

    Without any prior information on $A$, $B$, and $C$ except for the state dimension $\dimx$, use \eqref{eq:dataset} to design a linear dynamic output-feedback controller of the form
    \begin{equation}\label{eq:ctrl}
        \dot{x}_{\textup{c}} = A_{\textup{c}} x_{\textup{c}} + B_{\textup{c}} y, \quad u = C_{\textup{c}} x_{\textup{c}} + D_{\textup{c}} y,
    \end{equation}
    having state $x_{\textup{c}}$ and matrices $A_{\textup{c}}, B_{\textup{c}}, C_{\textup{c}}$ and $D_{\textup{c}}$, such that the origin $(x,x_{\textup{c}}) = 0$ of the interconnection of \eqref{eq:dot_x} and \eqref{eq:ctrl} is globally exponentially stable (GES).
\end{problem}

\section{Proposed Approach}
\label{sec:3}

\subsection{Kreisselmeier's Adaptive Filter}

Our approach is inspired by and built on Kreisselmeier's adaptive filter \cite{KRE79}, which is given by the following dynamical system:
\begin{equation}
\label{eq:dot_M}
    \dot{M} = FM + \begmat{y^\top \otimes I_{\dimx} &  u^\top \otimes I_{\dimx}},
\end{equation}
where $M \in \rea^{\dimx \times \colM}$ is the filter state and $F \in \rea^{\dimx \times \dimx}$ is a Hurwitz matrix, with $\mu \coloneqq (p+m)n$.
In \cite{KRE79}, it is shown that the filtered signal $M(t)$ satisfies an algebraic relation with the solution $x(t)$ of system \eqref{eq:dot_x}, for any initial condition of \eqref{eq:dot_x} and \eqref{eq:dot_M}.

Below, we provide a version of this result tailored to our case.

\begin{lemma}
\label{lem:Kreisselmeier}\it
Consider systems \eqref{eq:dot_x}, \eqref{eq:dot_M} and let Assumption \ref{hyp:minimal} hold.
Also, let $F$ in \eqref{eq:dot_M} be a Hurwitz matrix with distinct eigenvalues.
Then, there exist a non-singular matrix $T \in \rea^{\dimx \times \dimx}$ and a column vector $\theta \in \mathbb{R}^{\mu}$ such that, for all initial conditions $x(0) = x_0 \in \rea^n$ and $M(0) = M_0 \in \rea^{\dimx \times \colM}$,
\begin{equation}
\label{eq:Tx_M}
   T x(t) = M(t) \theta +   e^{Ft}(Tx_0 - M_0 \theta).
\end{equation}
\end{lemma}

\begin{proof}
Since the pair $(C, A)$ is observable due to Assumption \ref{hyp:minimal}, there exists a matrix $L$ such that the spectra of $A-LC$ and $F$ coincide.
Then, the fact that $F$ has distinct eigenvalues implies that $A - LC$ and $F$ are similar, i.e., there exists a non-singular matrix $T$ such that $F = T (A-LC) T^{-1}$. 
Using the change of coordinates $\bar x = T x$, system \eqref{eq:dot_x} with initial condition $x_0$ becomes
\begin{equation}\label{eq:syst_bar}
\begin{aligned}
    \dot{\bar x} = \bar A\bar x + \bar Bu, \qquad 
    y = \bar C \bar x
\end{aligned}
\end{equation}
with $\bar A= TAT^{-1}, \bar B= TB$, $\bar C=  CT^{-1}$, and initial condition $\bar x_0 = T x_0$.
System \eqref{eq:syst_bar} can be equivalently written as 
\begin{equation}
\label{eq:bar_x}
    \dot{\bar x} = F \bar x + TL y +  TB u.
\end{equation}

The main results in \cite{KRE79} follow immediately.
The solution $\bar{x}(t) = Tx(t)$ of \eqref{eq:bar_x} reads as
\begin{equation*}
\begin{aligned}
    Tx(t) & = e^{Ft} \bar x_0 + \int_0^t \! e^{F(t-s)} [TLy(s) + TBu(s)] \text{d}s
    \\
    & = e^{Ft} \bar x_0 + \int_0^t \! e^{F(t-s)} \!
    \begmat{y^\top(s) \otimes I_{\dimx} & u^\top(s) \otimes I_{\dimx}}\!
    \text{d}s \theta
    \\
    & = e^{Ft} T x_0 + \left( M(t) - e^{Ft}M_0\right)\theta,
\end{aligned}
\end{equation*}
where in the second line we exploited the identities $TLy = (y^\top \otimes I_{\dimx})\vec(TL)$, $TBu = (u^\top \otimes I_{\dimx})\vec(TB)$ and defined:
$$
\theta \coloneqq \begmat{\vec(TL) \\ \vec(TB)},
$$
while in the third line we used the solution $M(t)$ of \eqref{eq:dot_M}.
Rearranging the terms, we obtain~\eqref{eq:Tx_M}.
\end{proof}

\begin{remark}\it
While vectorizing the matrix variable \( M \) reveals that the filter dynamics of \( \mathrm{vec}(M) \) resemble those in \cite[Eq.~(85)]{BOSetal25}, the filter \eqref{eq:dot_M} does not exploit any structural property of MIMO systems.
As a consequence, unlike the design of \cite{BOSetal25}, our approach does not require knowledge of the observability indices of the pair $(C, A)$.
The trade-off, however, is the introduction of excessive coordinates in the filter.
This is not a major concern for offline data processing.
\end{remark}

\begin{remark}\it

The filter \eqref{eq:dot_M} does not rely on specific values of $A,B$, and $C$, but only on their dimensions for the given plant \eqref{eq:dot_x}.
Even in such a case, we have an algebraic relation \eqref{eq:Tx_M} between the unknown state $x$ and the available filtered variable $M$, where the vector $\theta$ is unknown.
Indeed, the use of similar filters to process input--output signals has a longstanding history in adaptive control \cite{ORTetal19}.
\end{remark}

\subsection{Filter-Based Output Feedback Stabilization}
\label{sec:31}

Our approach is motivated by the indirect output-feedback approach \cite{ANDPRA}.
In the indirect approach, we first design an observer that achieves an exponentially convergent estimate to the system state, uniformly in the ``exogenous'' signals $(u,y)$.
Then, we find a ``state-feedback'' controller for the observer.
Roughly speaking, we aim at the following property:
\begin{equation*}
\lim_{t\to\infty} M(t) = 0 \quad \mbox{(exp.)},
\end{equation*}
which then implies from \eqref{eq:Tx_M} that $x(t) \to 0$ exponentially as $t\to\infty$. 

From Lemma \ref{lem:Kreisselmeier}, we observe that the set 
\begin{equation}
\label{eq:cala}
  \mathcal{A} \coloneqq \{(x, M) \in \rea^{n}\times\rea^{n \times \colM}:  T x = M \theta \}
\end{equation}
is invariant and attractive. As a consequence, it makes sense to design a feedback law that depends only on $M$, as in classical controller design based on Luenberger observers.
However, it is important to remember that the system does not start in general on $\mathcal{A}$ because $M(0)$ would require unavailable prior information.
Thus, we need to account for the perturbation of the \emph{transverse coordinate}
\begin{equation}
\label{eq:beta}
\eps:=  Tx - M\theta \in \rea^n.
\end{equation}
To facilitate the data-driven controller design, we provide the following corollary of Lemma \ref{lem:Kreisselmeier}, which highlights the dynamics of the transverse coordinate $\eps$.
\begin{corollary}
\it \label{lem:transverse}
The transverse coordinate $\eps \in \rea^n$, defined in \eqref{eq:beta}, satisfies the dynamics  
\begin{equation}
\label{dot:beta}
    \begin{aligned}
\dot \eps = F \eps.
    \end{aligned}
\end{equation}
whose origin $\beta = 0$ is globally exponentially stable.
\end{corollary}
\begin{proof}
Rearranging terms in \eqref{eq:Tx_M}, we immediately obtain \eqref{dot:beta}, where $F$ is Hurwitz by design.
\end{proof}

\begin{remark}\it 
Transverse coordinates have been widely explored in various fields, especially in nonlinear control, such as path following of mechanical systems \cite{HLAetal}, orbital stabilization \cite{YIetal}, transverse stability analysis \cite{ANDetal}, and immersion and invariance control \cite{ASTORT}.
\end{remark}

\section{Canonical Non-Minimal Realization for Data-Driven Control}

We now show that the filter \eqref{eq:dot_M} can be interpreted as an exponentially convergent observer of a \emph{canonical non-minimal realization} of the plant \eqref{eq:dot_x}, of which we recall the definition.
\begin{definition}[\cite{BOSetal25}]\label{def:canonical}
\it
Under Assumption \ref{hyp:minimal}, a non-minimal realization of system \eqref{eq:dot_x} of the form
\begin{equation*}
 \dot\xi = A_\xi \xi + B_\xi u, \quad y = C_\xi \xi,
\end{equation*}
with $\xi \in \rea^{n_\xi}$ and $n_\xi >n$, is said to be \emph{canonical} if there exist a Hurwitz matrix $F_\xi \in \rea^{n_\xi \times n_\xi}$ and a matrix $L_\xi \in \rea^{n_\xi \times \dimy}$ such that $A_\xi = F_\xi + L_\xi C_\xi$.
In particular, $x = \Pi \xi$, where $\Pi$ is the solution to the equations
\begin{equation}
\label{eq:non_min_equation}
\begin{aligned}
    \Pi(F_\xi + L_\xi C_\xi) & = A \Pi, \quad \Pi B_\xi = B
    \\
    C_\xi & = C \Pi.
\end{aligned}
\end{equation}
\end{definition}
To highlight that the dynamics of $(x, M)$ restricted onto the set $\mathcal{A}$ is a canonical non-minimal realization of \eqref{eq:dot_x}, we first provide the vectorization of the filter dynamics \eqref{eq:dot_M}.
\begin{lemma}
\it \label{lem:vec}
Let $\hat \xi:= \mathrm{vec}(M)$.
Then, it holds that
\begin{equation}
\label{dot:hat_xi}
    \dot{\hat \xi} = F_\xi \hat \xi + B_{\xi} u + L_{\xi} y
\end{equation}
where
\begin{equation}
\begin{aligned}
    F_\xi & 
    \coloneqq I_{\colM} \otimes F
    \\
    B_\xi &
    \coloneqq 
    \begmat{0_{\,\dimy \dimx^2\times \dimu}\\
    I_{\dimu}\otimes \mathrm{vec}(I_{\dimx})
 }
    \\
    L_\xi &
    \coloneqq
    \begmat{
    I_{\dimy}\otimes \mathrm{vec}(I_{\dimx})\\
    0_{\dimu \dimx^2 \times \dimy}
    }.
\end{aligned}
\end{equation}
\end{lemma}


\begin{proof}
 The filter \eqref{eq:dot_M} can be written as 
\begin{equation}
\label{eq:dot_M2}
\dot{M} = FM + 
     \begmat{0_{\dimx\times \dimy\dimx} & u^{\top} \otimes I_{\dimx} } + \begmat{y^{\top}\otimes I_{\dimx} & 0_{\dimx\times\dimu\dimx}}.
\end{equation} 
We now perform vectorization of the terms in the right-hand side of \eqref{eq:dot_M2}.
For the first term, we obtain
\begin{equation}
\label{eq:vec2}
\begin{aligned}
    \operatorname{vec}(FM) & = (I_{\colM}\otimes F)\vec(M) = F_{\xi}\vec(M)
\end{aligned}
\end{equation}
For the second term in \eqref{eq:dot_M2}, it holds that
\begin{equation*}
\vec \left(\begmat{0_{\,\dimx\times \dimy\dimx} & u^{\top} \otimes I_{\dimx} } \right)
=
\vec\left(\left(\begmat{      0_{\dimy \times \dimu}\\
      I_{\dimu}} u\right)^\top \otimes I_{\dimx}
\right).
\end{equation*}
For convenience, let $a \coloneqq (u^\top [0_{\dimu \times \dimy}\;\;
      I_{\dimu}]^\top)^\top$ and use the relation $
\vec\bigl(\mathrm{e}_i^{\top}\otimes I_{\dimx}\bigr)
     = \mathrm{e}_i\otimes\operatorname{vec}I_{\dimx},
$
where $\mathrm{e}_i \in \mathbb{R}^{\dimy + \dimu}$ denotes the $i$-th component of the standard basis of $\rea^{\dimy + \dimu}$, and $a_i$ denotes the $i$-th entry of $a$. We can rewrite \eqref{eq:vec2} as
\begin{equation*}
\begin{aligned}
           \operatorname{vec}(a^{\top}\otimes I_{\dimx})
=
& \vec\left(\sum_{i=1}^{\dimy + \dimu}a_i\bigl(\mathrm{e}_i^{\top}\otimes I_{\dimx}\bigr)\right) 
\\
=&\sum_{i=1}^{\dimy + \dimu} \vec(\mathrm{e}_i^\top \otimes I_n)a_i 
\\=& \sum_{i=1}^{\dimy + \dimu} \bigl(\mathrm{e}_i \otimes \vec(I_n)\bigr)a_i
\\=& \bigl(I_{\dimy + \dimu}\otimes\operatorname{vec}(I_{\dimx})\bigr) a
\\
= & \bigl(I_{\dimy + \dimu}\otimes\operatorname{vec}(I_{\dimx})\bigr)\begmat{0_{\dimy \times \dimu}\\
      I_{\dimu}}u. \\
 =     &  
\begin{bmatrix}
0_{\,\dimy \dimx^2\times \dimu}\\
I_{\dimu}\otimes \mathrm{vec}(I_{\dimx})
\end{bmatrix}u = B_{\xi}u.
\end{aligned}
\end{equation*}

Using similar arguments, the third term in \eqref{eq:dot_M2} becomes
\begin{equation*}
    \vec
    \left(
    \begmat{y^{\top}\otimes I_{\dimx} & 0_{\dimx \times \dimu \dimx}}
    \right) 
    = \begin{bmatrix}
    I_{\dimy}\otimes \operatorname{vec}(I_{\dimx})\\
    0_{\dimu \dimx^2\times \dimy}.
    \end{bmatrix} y = 
   L_\xi y.
\end{equation*}
Combining the above computations, the vectorization of the dynamics \eqref{eq:dot_M} is given by
\begin{equation}
        \vec(\dot{M})
        =
        F_\xi \vec(M)
        + 
        B_\xi u
        +
        L_\xi y,
\end{equation}
which is exactly \eqref{dot:hat_xi}.
\end{proof}
We are ready to show the structure of the canonical non-minimal realization of system \eqref{eq:dot_x} as per Definition \ref{def:canonical}.
\begin{lemma}
\it
\label{lem:stablizable}
Under Assumption \ref{hyp:minimal}, the system 
\begin{equation}
\label{dot:xi}
\begin{aligned}
\dot \xi &= A_\xi \xi + B_\xi u
\\
y_\xi &= C_{\xi}\xi
\end{aligned}
\end{equation}
with $\xi \in \rea^{n_\xi}$, $n_\xi \coloneqq\dimx^2(\dimy + \dimu)$, and
\begin{equation}
\label{eq:vec_matrix_param}
\begin{aligned}
    A_\xi  \coloneqq F_\xi + L_\xi C_\xi
    , \quad 
    C_\xi  \coloneqq C(\theta^{\top}\otimes T^{-1})
\end{aligned}
\end{equation}
is a canonical non-minimal realization of the plant \eqref{eq:dot_x}. Moreover, the pair $(A_\xi, B_\xi)$ is stabilizable. 
\end{lemma}

\begin{proof}
By Lemma \ref{lem:Kreisselmeier}, the set $\cala$ defined in \eqref{eq:cala} is an invariant subspace.
Hence, it holds that
\begin{equation}
\label{eq:IC}
  \left.
  \begin{aligned}
        Tx(0) &= M(0)\theta
        \\
        \hat \xi(0) & = \vec(M(0))
  \end{aligned}
  \right\}
~\implies~ 
   x(t) = \Pi \hat\xi(t), \; \forall t\ge 0
\end{equation}
for any $u(t)$, where $\Pi$ is a full-row rank matrix defined as
\begin{equation}
    \Pi \coloneqq \theta^\top \otimes T^{-1}.
\end{equation}
From the definition of $C_\xi$ in \eqref{eq:vec_matrix_param}, we get
\begin{equation}
\label{eq:C}
C_\xi =C\Pi.
\end{equation}
Furthermore, on the set $\mathcal{A}$, we can compute the derivative of the right-hand side of \eqref{eq:IC} to obtain         
\begin{equation}\label{eq:dynamics_on_A}
    \begin{aligned}
        & Ax +Bu = \Pi (F_\xi \hat \xi + B_\xi u + L_\xi y) \\
        \implies \quad & A\Pi \hat \xi + Bu = \Pi(F_\xi \hat \xi + B_\xi u + L_\xi C \Pi\hat \xi) \\
        \implies \quad & A\Pi \hat\xi +Bu =\Pi A_\xi\hat\xi + \Pi B_\xi u.
    \end{aligned}
\end{equation}
Note that \eqref{eq:dynamics_on_A} holds for all $\col(\hat{\xi}, u) \in \rea^{n_\xi + \dimu}$.
In fact, since $\Pi$ has full-row rank, the set $\mathcal{A}$ is a linear subspace of dimension $n_{\xi}$ that corresponds, from $\hat{\xi} = \vec (M)$, to the graph of the map $\hat{\xi} \mapsto \Pi \hat{\xi}$, with free parameter $\hat{\xi} \in \rea^{n_\xi}$.
We thus conclude that the following equations hold:
\begin{equation}
\label{eq:AB}
A\Pi = \Pi A_{\xi}, \quad 
B = \Pi B_{\xi}.
\end{equation}
Finally, $A_\xi$ can be decomposed as $A_\xi = F_\xi + L_\xi C_\xi$, with $F_\xi$ Hurwitz.
We have thus obtained \eqref{eq:non_min_equation}, which proves that \eqref{eq:vec_matrix_param} is a canonical non-minimal realization of the plant \eqref{eq:dot_x}.
We conclude the proof by noting that the stabilizability of $(A_\xi, B_\xi)$ follows directly from \cite[Corollary 1]{BOSetal25}.
\end{proof}

\begin{remark}
\it 
If the initial condition of the non-minimal realization \eqref{dot:xi} satisfies $\Pi \xi(0) = x(0)$, then the filter \eqref{dot:hat_xi} in the vectorized form is an exponentially convergent observer of $\xi$, i.e. $\lim_{t\to\infty}|\hat \xi(t) - \xi(t)|=0$.
To see this, define the error variable $e_\xi \coloneqq \hat \xi - \xi$, whose dynamics are given by $\dot e_\xi = F_\xi e_\xi + L_\xi(Cx - C_\xi \xi)$.
From this special condition, we have $Cx= C_\xi \xi$.
Combining with the fact that $F_\xi$ is Hurwitz, we conclude the exponential stability of the dynamics of $e_\xi$.
\end{remark}

\begin{example}
\label{ex:1}\it
Consider a three-dimensional system with two inputs, two outputs
($\dimx = 3$, $\dimy = \dimu = 2$) and matrices
\[
A = \begin{bmatrix}
1 & 2 & 0 \\
0 & 2 & 1 \\
3 & 0 & 1
\end{bmatrix}, \quad
B = \begin{bmatrix}
1 & 0 \\
0 & 1 \\
1 & 2
\end{bmatrix}, \quad
C = \begin{bmatrix}
1 & 0 & 2 \\
0 & 1 & 1
\end{bmatrix}.
\]
We select $F = \operatorname{diag}(-20,-36,-45)$. Following the proofs of
Lemma~\ref{lem:Kreisselmeier}, we obtain
\[
T =
\begin{bmatrix}
6.1083 & -0.2340 & 17.5900 \\
-5.5982  & -2.9025  & -17.1889   \\
-1.4337 & 2.8633 & -0.0868
\end{bmatrix}
\]
and the corresponding column vector $\theta=[\theta_1, \theta_2]^\top$ with
$\theta_1= [$\rm{181.0436, $-$258.7007, $-$59.0428, 7.0687, $-$121.4904, 117.3910}$]$
\it and 
$\theta_2=[$\rm{23.6983, $-$22.7871, $-$1.5205, 34.9460, $-$37.2803, 2.6897}$]$. \it By the result of Lemmas~\ref{lem:vec} and~\ref{lem:stablizable}, we obtain
$A_{\xi}$ and $B_{\xi}$. The matrix $A_{\xi}$ has three positive eigenvalues
given by $\{3.2188,\, 0.3906 \pm 1.5274\mathrm{i}\}$, corresponding to the
eigenvalues of $A$. For all of these eigenvalues, the PBH test yields
$\operatorname{rank}[\lambda_i I - A_{\xi} \mid B_{\xi}] = 3$, which implies
that the pair $(A_{\xi},B_{\xi})$ is stabilizable. Note that the observability
indices of this system are $\nu_1 = 2$ and $\nu_2 = 1$, which are not uniform,
and thus the filter structure of \cite{BOSetal25} is not applicable.
\end{example}

\section{Data-driven Output-Feedback Stabilization}
\label{sec:34}

In this section, we exploit the previous technical results to design an output-feedback stabilizing controller from the dataset $\cale_{\tt D}$ in \eqref{eq:dataset}, under some excitation conditions of the plant \eqref{eq:dot_x}.

\subsection{Dynamics Compatible with the Input--Output Data}

Given $\cale_{\tt D}$, we can post-process the available trajectories $(u(t),y(t))$ over the interval $[0,t_{\tt D}]$ with the filter \eqref{eq:dot_M}.
By Lemma \ref{lem:stablizable}, to make the system \eqref{eq:dot_M} (or, equivalently, \eqref{dot:hat_xi}) describe the same input--output response as \eqref{eq:dot_x}, we would need to impose the constraint on the initial condition
\begin{equation}
    \Pi \vec(M(0)) = \Pi\hat{\xi}(0) = x_0,
\end{equation}
which would however require knowledge of $x_0$ and the matrices $A$, $B$, and $C$.

As a consequence, instead, we choose for simplicity the initial condition
\begin{equation}\label{eq:M_0}
    M(0) = 0_{n \times \colM},
\end{equation}
thus obtaining
\begin{equation}\label{eq:estimation_error}
\begin{aligned}
x = T^{-1} (M\theta  + \eps) = \Pi \hat \xi + T^{-1} \eps.
\end{aligned}
\end{equation}
We can then rewrite the interconnection of \eqref{eq:dot_x} and \eqref{eq:dot_M} in the coordinates $(\beta, \hat{\xi})$:
\begin{equation}
\label{dot:beta_hat_xi}
\begin{aligned}
    \dot \eps & = F \eps 
    \\
    \dot{\hat \xi} & 
    =
    A_\xi \hat \xi + B_\xi u + L_\xi(Cx - C_\xi \hat\xi)
    \\
    & = 
 A_\xi \hat \xi + B_\xi u + L_\xi C(x - \Pi \hat \xi)
    \\
    & =
     A_\xi \hat \xi + L_\xi C T^{-1} \eps + B_\xi u,
\end{aligned}
\end{equation}
which is compatible with the dataset $\cale_{\tt D}$ in \eqref{eq:dataset} by letting $\hat{\xi}(0) = 0_{n \mu}$ as per \eqref{eq:M_0} and $\beta(0) = Tx(0)$ due to \eqref{eq:estimation_error}.
As a consequence, by processing the dataset with the filter \eqref{eq:dot_M}, we have the access to the data of $M$ and $\dot M$ (equivalently, $\hat \xi$ and $\dot{\hat{\xi}}$)\footnote{Hereafter, when we say that a derivative is available, we mean that the right-hand side of the corresponding differential equation (e.g., the implemented filter \eqref{eq:dot_M}) is given by available signals and thus can be sampled.}, while $\eps(t)$ is not available because its initial condition is unknown.
We select the matrix $F$ satisfying the following.\footnote{For ease of presentation, we impose the following condition of the user-selected matrix $F$. It can be extended to the case of complex eigenvalues.} 

\begin{condition}
\it\label{cond:1}
The eigenvalues of the Hurwitz matrix $F$ are distinct and real. 
\end{condition}

Then, we can diagonalize the matrix $F$ as $F = T_F \Lambda T_F^{-1}$ for some invertible $T_F$ and $\Lambda = \diag(\lambda_1, \ldots, \lambda_n)$, where $\lambda_j$ ($j \in \{ 1, \ldots, n\}$) are the eigenvalues of $F$.
Thus, the unavailable $\eps$ can be represented as
\begin{equation}
\begin{aligned}
    \eps(t) & = T_F  e^{\Lambda t} T_F^{-1} \eps(0)
    =
    \Gamma \chi(t),
\end{aligned}
\end{equation}
for some matrix $\Gamma$ and $\chi(t) \in \rea^n$ that satisfies the following differential equation, which can be implemented together with the filters \eqref{eq:dot_M}:
\begin{equation}\label{eq:dot_eps}
\begin{aligned}
    \dot\chi(t) & = \Lambda \chi(t), \quad \chi(0)= \begin{bmatrix}
        1, \dots,1
    \end{bmatrix}^{\top} \in \rea ^{n}.
\end{aligned}
\end{equation}
As a consequence, by letting
$$
z \coloneqq \begmat{\chi \\ \hat \xi},
$$
we obtain the following dynamics:
\begin{equation}
\label{eq:composite_sys}
\dot z =
\underbrace{\begin{bmatrix}
        \Lambda & 0\\
        L_\xi C T^{-1} \Gamma & A_\xi
    \end{bmatrix}
    }_{A_e}
    z 
    +
        \underbrace{\begin{bmatrix}
        0 
        \\
        B_{\xi}
    \end{bmatrix}
    }_{B_e}
    u,
\end{equation}
which are compatible with the dataset $\cale_{\tt D}$ in \eqref{eq:dataset} by letting $\hat{\xi}(0) = 0_{n\mu}$, $\chi(0) = [1\; \cdots \;1]^\top$, and where $z$, $\dot{z}$, and $u$ are available for measurement, without any need for computing dirty derivatives.
Since $(A_\xi, B_\xi)$ is stabilizable and $\Lambda$ is Hurwitz, the extended pair $(A_e, B_e)$ in system \eqref{eq:composite_sys} is stabilizable. 

We obtain the following result, providing the desired class of controllers for addressing Problem \ref{problem:1}.

\begin{lemma}
\label{prop:stab_filt_stab_sys}
{\em Let Assumption \ref{hyp:minimal} hold and consider system \eqref{eq:composite_sys}.
For any matrix $K_e$ making $A_e + B_e K_e$ Hurwitz, the output-feedback controller}
\begin{equation}
\label{eq:u1}
\begin{aligned}
\dot{M} &= FM
+ \begmat{y^\top \otimes I_n &  u^\top \otimes I_n }, 
\quad M(0)=0_{n\times \mu}
\\
u & = K_e \begmat{0_n \\ \mbox{vec}(M)}
\end{aligned}
\end{equation}
{\em is such that the origin $(x, M) = 0$ of the closed-loop interconnection of \eqref{eq:dot_x} and \eqref{eq:u1} is globally exponentially stable.}
\end{lemma}
\begin{proof}
Consider the feedback law
\begin{equation}
\label{eq:u2}
u = K_e z,
\end{equation}
under which the system \eqref{eq:composite_sys} is GES. From \eqref{dot:xi}, we know that the closed loop is GES under \eqref{eq:u2}. 
Comparing the control laws \eqref{eq:u2} and \eqref{eq:u1}, the difference is an exponentially decaying signal injected into the closed-loop dynamics.
The proof follows from standard perturbation arguments.
\end{proof}

Given the dataset $\cale_{\tt D}$ and the trajectories of $z(t)$ generated by implementing the filters \eqref{eq:dot_M} and \eqref{eq:dot_eps}, the following data batches can be obtained:

\begin{equation}
\label{data_batch}
\begin{aligned}
     U \; & :=\; \begmat{u(t_1) & u(t_2) & \dots & u(t_N)}\in\mathbb R^{m\times N}
     \\
     Z \; & := \; \begmat{z(t_1) & z(t_2) & \dots & z(t_N)}\in\mathbb R^{n_z\times N}
     \\
     \dot Z \; & := \; \begmat{\dot z(t_1) & \dot z(t_2) & \dots & \dot z(t_N)}\in\mathbb R^{n_z\times N},
\end{aligned}
\end{equation}
where $n_z:=n+ n_\xi$ and $0 \leq t_1< t_2< \ldots < t_N \leq t_{\tt D}$ are the sampling instants.

By Lemma \ref{prop:stab_filt_stab_sys}, Problem \ref{problem:1} is solved if, using \eqref{data_batch}, we can find $K_e$ such that $A_e + B_e K_e$ is Hurwitz.

\subsection{Data-driven State Decomposition for Stabilizable Systems}

Given the dynamics \eqref{eq:composite_sys} and the batches \eqref{data_batch}, one may be tempted to directly apply the LMIs for data-driven stabilization by state feedback of \cite{DEPTES} to find a stabilizing gain $K_e$.
In this context, it is known that a sufficient condition for the feasibility of the LMIs would be
\begin{equation}
\label{cond:full_rank1}
    \rank \begmat{U \\ Z} = \dimu + \dimx + n_\xi.
\end{equation}
In our scenario, in general, \eqref{cond:full_rank1} cannot be guaranteed, regardless of the input $u(t)$ injected in the system.
In fact, since $(A_e, B_e)$ is stabilizable but not controllable and due to the high number of states introduced by the filter \eqref{eq:dot_M}, it is possible that the uncontrollable modes contained in $z(t)$ are linearly dependent, so that the following matrix
\begin{equation}\label{eq:Z}
\mathcal{Z} \coloneqq \int_{0}^{t_{\tt D}} z(s) z(s)^\top \text{d}s,
\end{equation}
loses rank.
In turn, $\rank \mathcal{Z} < n + n_\xi$ implies $\rank Z < n + n_\xi$, regardless of the sampling choice, causing the LMIs in \cite{DEPTES} to become unfeasible because a necessary rank condition fails, see \cite[Proof of Thm. 1]{BOSetal24}.

To circumvent this issue, in this work, we propose an approach that handles the scenario $\rank \mathcal{Z} < n + n_\xi$ by directly exploiting the stabilizability of the pair $(A_e, B_e)$.
Notably, similar ideas are under investigation also in the discrete-time context \cite{shakouri2025data}.

In particular, our key idea is to require that $\mathcal{Z}$ loses rank only in directions corresponding to some uncontrollable modes of $z(t)$, and can be formalized as follows:
\begin{equation}\label{eq:Im_Z}
    \im \mathcal{R} \subseteq \im \mathcal{Z},
\end{equation}
where $\mathcal{R} \coloneqq [B_e \;\; A_eB_e\;\; \cdots \;\; A_e^{n_z - 1}B_e]$ is the reachability matrix of the pair $(A_e, B_e)$.
Note that condition \eqref{eq:Im_Z} is weaker than $\mathcal{Z} \succ 0$ in the scenario where $\rank \mathcal{R} < n + n_\xi$.
Since $\mathcal{Z}$ can be computed from the available signal $z(t)$, it makes sense to exploit \eqref{eq:Im_Z} to search for a change of coordinates $T_z \in \rea^{n_z\times n_z}$, by which we divide $z$ as follows:
\begin{equation}\label{eq:z_ab}
    \begmat{z_a \\ z_b} = T_z z,
\end{equation}
where we aim to include all controllable modes in $z_a$\footnote{The partition $z_a$ may also possibly include uncontrollable modes.}. The approach to compute $T_z$ is described in Algorithm \ref{alg:Tz}.
\begin{algorithm}[ht]
\caption{Searching for the Transformation $T_z$}
\label{alg:Tz}
\begin{algorithmic}[1]
\Require Signal $z(t):[0,t_{\tt D}]\to\mathbb{R}^{n_z}$.
\Ensure Change of coordinates $T_z \in \mathbb{R}^{n_z\times n_z}$.
\State $\mathcal{Z} \gets 
\int_0^{t_{\tt D}} z(s) z(s)^\top \mathrm{d}s \in\mathbb{R}^{n_z\times n_z}$ 
\State Compute the singular value decomposition: $\mathcal{Z} = \mathcal{U}\Sigma \mathcal{U}^\top$ with $\mathcal{U}$ orthogonal, $\Sigma:= \diag(\Sigma_0, 0, \ldots,0)$, and $\Sigma_0:=\diag(\lambda_1, \ldots, \lambda_l)$, where $l \coloneqq 
\rank \mathcal{Z}$.

\Comment{$\lambda_1 \geq \cdots \geq \lambda_l > 0$.}
\State \Return $T_z \gets \mathcal{U}^\top$.
\end{algorithmic}
\end{algorithm}

Conformally, we partition the system matrix and the input matrix as
\begin{equation}
\label{eq:transformed_dynamic}
\begmat{\dot z_a \\ \dot z_b} = \begmat{A_a & A_{ab} \\ A_{ba} & A_b}\begmat{z_a \\ z_b} + \begmat{B_a \\ B_b}u.
\end{equation}
The following result ensures that $T_z$ from Algorithm \ref{alg:Tz} provides the decomposition we aim for.
\begin{lemma}
\label{lem:AB}\it
Suppose that the solution $z(t)$ satisfies condition \eqref{eq:Im_Z}.
Then, the similarity transformation $T_z\in \rea^{n_{z}\times n_{z}}$ obtained from Algorithm \ref{alg:Tz} is such that system \eqref{eq:transformed_dynamic} satisfies $A_{ba} = 0$, $B_b = 0$, and $A_b$ Hurwitz.
\end{lemma}
\begin{proof}
From \eqref{eq:Im_Z}, we obtain that
\begin{equation}
    \ker \mathcal{Z}^\top \subseteq \ker \mathcal{R}^\top.
\end{equation}
The matrix $T_z$ obtained from Algorithm \ref{alg:Tz} can be split as
\begin{equation}
    T_z = \begmat{T_a \\ T_b},
\end{equation}
where in particular $T_b \in \rea^{(n_z - l) \times n_z}$ is associated with $z_b$.
Note that any row $v^\top$ of $T_b$ satisfies $v^\top \mathcal{Z} = 0$, i.e., $v \in \ker \mathcal{Z}^\top$.
Therefore, $v \in \ker \mathcal{R}^\top$.
It follows that $T_b \mathcal{R} = 0$ and, in particular,
\begin{equation}
    B_b = T_b  B_e = 0.
\end{equation}
The dynamics of $z_b$ becomes
\begin{equation}\label{eq:B_b_0}
    \dot{z}_b = A_b z_b + A_{ba}z_a.
\end{equation}
Applying the transformation $T_z$ to the matrix $\mathcal{Z}$ yields: 
$$
\begin{aligned}
T_z \mathcal{Z} T_z^\top = \int_{t}^{t_{\tt D}} \begmat{z_a(s)\\ z_b(s)} \begmat{z_a(s) \\ z_b(s)} \text{d}s  = 
\begmat{\Sigma_0 & 0 \\ 0 & 0}
\\
\implies
\int_0^{t_{\tt D}} z_b(s) z_b(s)^\top \text{d}s = 0\\
\implies \int_0^{t_{\tt D}} |z_b(s)|^2 \text{d}s = 0
\end{aligned}
$$
In particular, since $z(t)$ is an absolutely continuous signal, it necessarily holds that
\begin{equation}
\label{eq:z=0}
    z_b(t) = 0 ,\quad \text{for all } t\in [0, t_{\tt D}],
\end{equation}
and, furthermore, 
\begin{equation}
\label{eq:dotz=0}
    \dot{z}_b(t) = 0 ,\quad \text{for almost all } t\in [0, t_{\tt D}].
\end{equation}
From \eqref{eq:B_b_0}, we infer that $A_{ba} z_a(t) = 0$ for almost all $t \in [0, t_{\tt D}]$.
Since, similarly, $A_{ba}z_a(t)z_a^\top(t) = 0$, we conclude
\begin{equation}
    A_{ba}\int_{0}^{t_{\tt D}}z_a(s)z_a^\top(s)\mathrm{d}s = A_{ba}\Sigma_0 = 0 \iff A_{ba} = 0.
\end{equation}
Since we have proved that $A_{ba} = 0$ and $B_b = 0$, we obtain that $A_b$ is Hurwitz because $z_b$ is necessarily part of the uncontrollable subsystem, while the pair $(A_e, B_e)$ is stabilizable.
\end{proof}

\subsection{Data-driven Controller Synthesis}

\begin{algorithm}[!t]
\caption{Input/Output Data-driven Stabilization}
\label{alg:data-driven-stab}
\begin{algorithmic}[0]
\Require Dataset $\cale_{\tt D}:=\{u(t),y(t)\}_{[0,t_{\tt D}]}$.

\Statex \textbf{Initial Setting} 
\State
\textit{Filter parameters:} $F \in \mathbb{R}^{n\times n}$ Hurwitz and with $n$ distinct real eigenvalues.
\State
\textit{Number of samples:}  $N \ge 1$.

\Statex \textbf{Data Batches Construction}
\State
\textit{Filter of the data:} simulate for $t\in[0,t_{\tt D}]$:
\begin{equation*}
    \dot{M}(t) \!=\! FM(t) + \begmat{y^\top\!(t)\! \otimes \! I_{\dimx}\! & \! u^\top\!(t)\! \otimes \! I_{\dimx}}\!,\;\; M(0)\!=\!0_{n\times \colM}.
\end{equation*}
\State \textit{Auxiliary dynamics:} simulate for $t\in[0,t_{\tt D}]$:
\begin{equation*}
    \dot\chi  = \Lambda \chi, \qquad \chi(0)= \col(1,\ldots, 1). 
\end{equation*}
\State \textit{Filtered signal:} $z(t) \coloneqq (\chi(t), \vec(M(t))).$
\State 
\textit{State decomposition}: Run \textbf{Algorithm~\ref{alg:Tz}} and compute the signals $z_a(t)$, $\dot{z}_a(t)$ using $T_z$ as in \eqref{eq:z_ab}.

\State \textit{Data batches:} let $0 \leq t_1< t_2< \ldots < t_N \leq t_{\tt D}$ and
\begin{equation}\label{eq:batches}
    \begin{aligned}
        Z_a &\coloneqq \begmat{z_a(t_1) & z_a(t_2) & \ldots & z_a(t_N)} \in \rea^{l \times N}\\
        \dot{Z}_a &\coloneqq \begmat{\dot{z}_a(t_1) & \dot{z}_a(t_2) & \ldots & \dot{z}_a(t_N)} \in \rea^{l \times N}\\
    U&\coloneqq \begmat{\;u(t_1) & \;u(t_2) & \;\dots & \;u(t_N)}\in\mathbb R^{m\times N}
    \end{aligned}
\end{equation}

\Comment{$t_1,\ldots, t_N$ such that $\rank\begmat{Z_a\\ U} = l + m$.
}

\Statex \textbf{Stabilizing Gain Computation}
\State \textit{LMI:} find $Q\in\rea^{N\times l}$ such that
\begin{equation}
\label{eq:opt_LMI}
 \left\{
\begin{aligned}
& Z_aQ = Q^{\top}Z_a^{\top} \succ 0,\\
&\dot{Z}_a\,Q + Q^\top \dot{Z}_a^\top \prec 0.
\end{aligned}
\right.    
\end{equation}

\State \textit{Gain matrix:} \begin{equation}
\label{eq:control_gain}
    K_e = \begmat{UQ(Z_aQ)^{-1} & 0_{m\times(n_z-l)}}T_z.
\end{equation}

\Ensure Output-feedback law: \eqref{eq:u1}
\end{algorithmic}
\end{algorithm}

We now design an LMI that computes $K_e$ as specified in Lemma \ref{prop:stab_filt_stab_sys}.
The overall procedure is summarized in Algorithm \ref{alg:data-driven-stab} and is performed under the following assumption, based on the results of the previous subsection.

\begin{assumption}\label{hyp:IE}
    \it
    The dataset $\cale_{\tt D}$ in \eqref{eq:dataset} and the filtered signal $z(t) = (\chi(t), \mathrm{vec}(M(t)))$, generated by \eqref{eq:dot_M} with $M(0) = 0_{n \times \mu}$ and \eqref{eq:dot_eps}, satisfy \eqref{eq:Im_Z} and the following \emph{interval excitation} condition:
    \begin{equation}
        \label{IE:1}
        \int_0^{t_{\tt D}} \begmat{z_a(s)\\ u(s)} \begmat{z_a(s)\\ u(s)}^\top \mathrm{d}s \succ 0.
    \end{equation}
\end{assumption}
Under Assumption \ref{hyp:IE}, by \cite[Proposition 1]{WANetal}, it holds that there exist sampling instants $0 \leq t_1< t_2< \ldots < t_N \leq t_{\tt D}$ such that the sampled data batches \eqref{eq:batches} satisfy
\begin{equation}
\label{rank:za}
\rank \begmat{Z_a \\ U} = m + l.
\end{equation}
With the resulting data batches, it is possible to exploit the LMIs for data-driven stabilization of \cite{DEPTES} in the reduced coordinates.
In particular, we can compute the control gain $K_e$ of \eqref{eq:u1} by solving the LMI \eqref{eq:opt_LMI} and then applying \eqref{eq:control_gain}.
The theoretical properties ensured by Algorithm \ref{alg:data-driven-stab} are summarized in the next theorem, which is the main result of this work.
\begin{theorem}
\it
Consider Algorithm \ref{alg:data-driven-stab} under Assumptions \ref{hyp:minimal} and \ref{hyp:IE}.
Then, the matrix $K_e$ computed from \eqref{eq:opt_LMI}, \eqref{eq:control_gain} ensures that the matrix $A_e + B_e K_e$ is Hurwitz, where $A_e$ and $B_e$ are given in \eqref{eq:composite_sys}.
As a consequence, the dynamic output-feedback controller \eqref{eq:u1} solves Problem \ref{problem:1}.
\end{theorem}

\begin{proof}
Using the change of coordinate \eqref{eq:transformed_dynamic}, the full rank condition \eqref{rank:za} implies the feasibility of the LMI \eqref{eq:opt_LMI}, so that by standard arguments \cite{DEPTES} the matrix $K =  UQ(Z_aQ)^{-1}$ ensures that $A_a + B_a K$ is Hurwitz.
As a consequence, the feedback gain $K_e$ in \eqref{eq:control_gain} ensures that
\begin{equation*}
    A_e + B_e K_e = T_z^{-1} \begmat{ A_a +B_a K & A_{ab} \\ 0 & A_b} T_z.
\end{equation*}
From Lemma \ref{lem:AB}, the matrix $A_b$ is Hurwitz.
Therefore, the closed-loop system matrix $(A_e+B_e K_e)$ is Hurwitz.
Invoking Lemma \ref{prop:stab_filt_stab_sys}, we complete the proof.
\end{proof}

\section{Numerical Results} 

In this section, we showcase our theoretical results by applying Algorithm \ref{alg:data-driven-stab} to the system of Example \ref{ex:1}.
All simulations are implemented in MATLAB using YALMIP \cite{Lofberg2004} and MOSEK \cite{mosek} to solve the LMI \eqref{eq:opt_LMI}.
The dataset $\cale_{\tt D}$ is generated by simulating the system in the interval $[0, 3]$ s by considering an input $u(t)$ that is the sum of four sinusoids for each channel.

We consider a periodic sampling with period $\Delta = 0.01$s and choose $F = \operatorname{diag}(-20,\ -36,\ -40)$. Using the proposed approach, the resulting control gain is $K_e = \begmat{ 0_{2\times 3} & K & 0_{2\times27} }$ with $K:=\col(K_1, K_2)$, $K_1 = [$6.7 , 2.3 , 10.5 , 9.3 , 55.6 , $-$52.9 , $-$157.2 , $-$15.6 , $-$16.1 , 150.6 , 307.7 , 0.3 $]$ and $K_2=[$7.9 , $-$14.2 , 3.4 , $-$19.1, 42.7 , 35.8 , $-$157.0 , $-$24.7 , 12.0 , 157.6 , 129.1 , 1.0$]$.

The spectrum of the closed-loop system matrix $A_e + B_e K_e$ is shown in Fig.~\ref{fig:specplot}, confirming the closed-loop stability.
Finally, an example of closed-loop response under the proposed output-feedback controller, initialized in $x(0) = \col(-1, 1, 2)$, is illustrated in Fig.~\ref{fig:twofigs1}.

\begin{figure}[!htbp]
    \centering
    \begin{subfigure}[b]{0.27\textwidth}
        \centering
        \centering
\begin{tikzpicture}
  \begin{axis}[
      xlabel={\large Time(s)},
      width=0.99\textwidth,
      height=0.8\textwidth,    
      xmin=0, 
      xmax=7.5,
      ymin=0,
      ymax=0.6,
      enlarge x limits=false
  ]
    \addplot[blue, line width=1.2pt] table {Data/M_data_1000000.dat};
    \node at (axis cs:3.5,0.2) {\large $\|M\|_F$};
  \end{axis}
\end{tikzpicture} 
        \label{fig:subfig1}
    \end{subfigure}
\hspace{-1.2cm}
    \begin{subfigure}[b]{0.27\textwidth}
        \centering
        \centering
\begin{tikzpicture}
  \begin{axis}[
      xlabel={\large Time (s)},
      width=0.99\textwidth,
      height=.8\textwidth,  
      xmin=0,
      xmax=10, 
      enlarge x limits=false,
      legend style={at={(0.98,0.98)},anchor=north east},
      legend cell align={left}
  ]
    \addplot[myblue, line width=1.2pt]  table[x index=0, y index=1] {Data/x_data_1000000.dat};
    \addlegendentry{$x_{1}$}
    \addplot[orange, line width=1.2pt]   table[x index=0, y index=2] {Data/x_data_1000000.dat};
    \addlegendentry{$x_{2}$}
    \addplot[mygreen, line width=1.2pt] table[x index=0, y index=3] {Data/x_data_1000000.dat};
    \addlegendentry{$x_{3}$}
  \end{axis}
\end{tikzpicture}
        \label{fig:subfig2}
    \end{subfigure}
    \caption{Simulation results of the closed-loop system}
    \label{fig:twofigs1}
\end{figure}
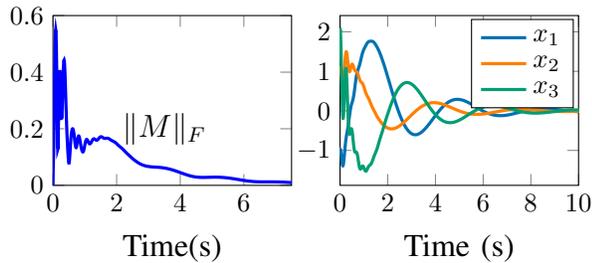

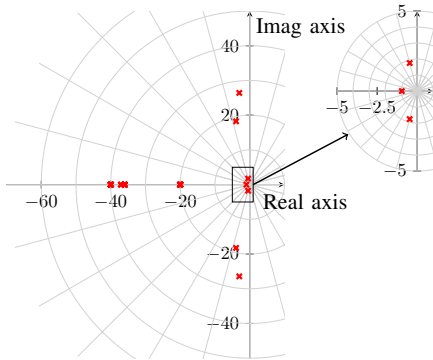
\begin{figure}[htbp]
    \centering
    \scalebox{0.7}{\begin{tikzpicture}
  \begin{axis}[
    name=main,
    width=0.62\textwidth, height=0.46\textwidth,
    axis lines=middle,
    axis equal image,
    xmin=-70, xmax=10, ymin=-50, ymax=50,
    enlargelimits=false, clip=true,
    tick style={black},
    every axis plot/.append style={line join=round},
    xlabel={\large Real axis},
    ylabel={\large Imag axis},
    xlabel style={at={(axis description cs:0.9,0.45)}, anchor=west},
  ]

    \pgfplotsinvokeforeach{10,20,30,40,50,60}{
      \addplot[domain=0:360, samples=361, thin, gray!35]
        ({#1*cos(x)},{#1*sin(x)});
    }

    \pgfplotsinvokeforeach{0,15,...,345}{
      \addplot[thin, gray!35, samples=2, domain=-70:10, variable=\r]
        ({\r*cos(#1)},{\r*sin(#1)});
    }

    \addplot[
      only marks, mark=x, mark size=2 pt, red, line width=1.2pt
    ] table[col sep=space, header=false]{Data/eig_points.dat};

    \draw[black, very thin]
      (axis cs:-5,-5) rectangle (axis cs:1,5);

    \coordinate (box_corner) at (axis cs:1,0);

  \end{axis}

  \begin{axis}[
    name=zoom,
    at={(main.north east)},
    anchor=north west,
    xshift=10mm,
    width=0.26\textwidth,
    height=0.26\textwidth,
    axis lines=middle,
    axis equal image,
    xmin=-5, xmax=1,
    ymin=-5, ymax=5,
    xtick={-5,-2.5,0},
    ytick={-5,0,5},
    tick style={black},
  ]

    \pgfplotsinvokeforeach{1,2,3,4,5}{
      \addplot[domain=0:360, samples=181, thin, gray!35]
        ({#1*cos(x)},{#1*sin(x)});
    }

    \pgfplotsinvokeforeach{0,15,...,345}{
      \addplot[thin, gray!35, samples=2, domain=-5:1, variable=\r]
        ({\r*cos(#1)},{\r*sin(#1)});
    }

    \addplot[
      only marks, mark=x, mark size=2 pt, red, line width=1.2pt
    ] table[col sep=space, header=false]{Data/eig_points.dat};

  \end{axis}

  \draw[->, thick, black]
    (box_corner) -- ([xshift = -0.7cm, yshift = -0.82cm]zoom.center);

\end{tikzpicture}}
    \caption{Spectrum of $A_e + B_eK_e$
    }
    \label{fig:specplot}
\end{figure}

\section{Concluding Remarks}
We proposed an approach for data-driven stabilization of general continuous-time LTI MIMO systems without any restriction on the observability indices.
We showed that Kreisselmeier's adaptive filter acts as an observer of a canonical non-minimal realization of the plant, allowing to design an output-feedback controller comprising the filter dynamics and a linear feedback law depending on the filter states.
Due the filter structure, we compute the gains of the feedback law with a data-based LMI formulation where, exploiting the stabilizability of the non-minimal realization, it is not necessary to have full-rank data.
Future work will aim at reducing the dimension of the controller dynamics by either considering smaller filter structures or investigating model reduction to simplify the controller implementation.
Furthermore, we will investigate how to remove the assumption of prior knowledge of the state dimension $n$.

\bibliographystyle{IEEEtran}
\bibliography{refs.bib}
\end{document}